%% file: ms.tex
\documentclass[11pt]{article}
\usepackage[left=1in, top=1in, right=1in, bottom=1in]{geometry}
\usepackage{indentfirst}
\usepackage{amsmath, amssymb, amsthm, amsfonts}
\usepackage{tikz, comment}
\usepackage{algpseudocode, algorithm}
\usepackage{tcolorbox, qtree}
\usepackage{enumerate}
\usepackage{fancyvrb}
\usepackage{caption}
\usepackage{xfrac}
\usepackage{bbm}
\usetikzlibrary{decorations.pathreplacing}
\newcommand{\citet}{\cite}

\newtheorem{theorem}{Theorem}

\newtheorem{observation}{Observation}

\newtheorem{define}{Definition}
\newtheorem{lemma}{Lemma}

\newtheorem{claim}{Claim}

\newtheorem{example}{Example}
\newcommand{\E}{\mathbf{E}}

\usepackage{bm}
\renewcommand{\vec}{\bm}

\title{Group Fairness in Committee Selection}
\date{\today}

\newcommand*\samethanks[1][\value{footnote}]{\footnotemark[#1]}

\author{
Yu Cheng\thanks{Department of Computer Science, Duke University. \tt{\{yucheng, kamesh, knwang\}@cs.duke.edu}}
\and
Zhihao Jiang\thanks{Institute for Interdisciplinary Information Sciences, Tsinghua University. \tt{jzh16@mails.tsinghua.edu.cn}}
\and
Kamesh Munagala\samethanks[1]
\and
Kangning Wang\samethanks[1]
}

\begin{document}
\maketitle

\input{abstract.tex}

\input{model.tex}

\input{related.tex}
\input{approval-proof.tex}
\input{runtime.tex}
\input{efficient.tex}
\input{counterexample.tex}
\input{conclusions.tex}

\section*{Acknowledgments}
Kamesh Munagala is supported by NSF grants CCF-1408784, CCF-1637397, and IIS-1447554; and research awards from Adobe and Facebook. Kangning Wang is supported by NSF grants IIS-1447554 and CCF-1637397. This work was done while Zhihao Jiang was visiting Duke University.

\bibliographystyle{acm}
\bibliography{abb,ultimate,refs}

\end{document}

%% file: abstract.tex
\begin{abstract}
In this paper, we study {\em fairness} in committee selection problems. We consider a general notion of fairness via stability: A committee is \emph{stable} if no coalition of voters can deviate and choose a committee of proportional size, so that all these voters strictly prefer the new committee to the existing one. Our main contribution is to extend this definition to stability of a distribution (or lottery) over committees. We consider two canonical voter preference models: the {\sc Approval Set} setting where each voter approves a set of candidates and prefers committees with larger intersection with this set; and the {\sc Ranking} setting where each voter ranks committees based on how much she likes her favorite candidate in a committee. Our main result is to show that stable lotteries always exist for these canonical preference models. 
Interestingly, given preferences of voters over committees, the procedure for computing an approximately stable lottery is the {\em same} for both models and therefore extends to the setting where some voters have the former preference structure and others have the latter. Our existence proof uses the probabilistic method and a new large deviation inequality that may be of independent interest. 
\end{abstract}

%% file: model.tex

\section{Introduction}
\label{sec:intro}
\newcommand{\N}{\mathcal{N}}
\newcommand{\C}{\mathcal{C}}
One of the central problems in social choice theory is {\em committee selection}, or {\em multi-winner elections}. In this problem, there is a set of voters (or agents) $\N = [n] = \{1, \ldots, n\}$, and a set of candidates $\C = [m]$. A committee is a subset of candidates, and the goal is to find a committee of given size $K$.  Committee selection arises in choosing a parliament, in making group hiring decisions, and in computer systems~\cite{ROBUS,Psomas}. A long line of recent literature~\cite{Procaccia2008,Meir2008,Lu2011,Brill,Sanchez,PJR2018} has studied the complexity and axiomatization of voting rules in this setting.

One classic objective in committee selection is {\em fairness} or {\em proportionality}: Every demographic of voters should feel that they have been fairly represented. They should not have the incentive to deviate and choose their own committee of proportionally smaller size which gives all of them higher utility. In the typical setting where these demographic slices are not known upfront, the notion of proportionality attempts to be fair to all subsets of voters. This general idea dates back more than a century~\cite{Droop}, and has recently received significant attention~\cite{Chamberlain,Monroe,Brams2007,Brill,Sanchez,PJR2018}. In fact, there are several elections, both at a group level and national level, that attempt to find committees (or parliaments) that provide approximately proportional representation. For instance, the popular Single Transferable Vote (STV) rule is used in parliamentary elections in Ireland and Australia, and in several municipal elections in the USA. This rule attempts to find a proportional solution.

Fairness in committee selection arises in many other applications outside of social choice as well. For example, consider a shared cache for data items in a multi-tenant cloud system, where each data item is used by several long-running applications~\cite{ROBUS,Psomas}. Each data item can be treated as a candidate, and each application as a voter whose utility for an item corresponds to the speedup obtained by caching that item. In this context, we need a {\em fair} caching policy that provides proportional speedup to all applications.

In this paper, we propose a new notion of proportionality in committee selection that generalizes several previously considered notions. Our main contribution is to show that stable solutions always exist for the stability notion that we propose, and such solutions can be computed efficiently. In contrast, stable solutions may not exist for some of the previously studied notions, and for the notions where stable solutions do exist, we do not know how to compute them efficiently (see Section~\ref{sec:related} for a detailed discussion).

\subsection{Preference Models}
Before proceeding further, we define the preference model of the voters for committees. We consider two canonical ordinal preference models over committees, both of which have been extensively studied in social choice literature.

\begin{description}
\item[{\sc Approval Set}.] In this model~\cite{Brams2007,Brill,Sanchez,PJR2018}, each voter $v$ specifies an {\em approval set} $A_v \subseteq \C$ of candidates. Given two committees $S_1$ and $S_2$, $S_1 \succ_v S_2$ iff $|S_1 \cap A_v| > |S_2 \cap A_v|$, {\em i.e.}, the voter strictly prefers committees in which she has more approved candidates.
\item[{\sc Ranking}.] In this model~\cite{ElkindFSS17}, each voter $v$ has a preference order over candidates in $\C$. In this case, $S_1 \succ_v S_2$ iff $v$'s favorite candidate in $S_1$ is ranked higher (in her preference ordering) than her favorite candidate in $S_2$. The Chamberlin-Courant voting rule~\citet{Chamberlain} for committee selection finds the social optimum assuming a cardinal preference function of this form.
\end{description}

These models have been extensively studied because it is relatively simple to elicit an approval vote or a ranking over candidates. Viewed in terms of underlying cardinal utility functions\footnote{Note that we do not elicit these cardinal utilities.}, both these models are special cases of {\em submodular utilities}: Voter $v$ has utility $u_v(S)$ for committee $S \subseteq \C$, where $u_v$ is a submodular set function. The voter prefers committees that give her larger utility.  The {\sc approval set} case corresponds to $u_v(S) = f_v(|S \cap A_v|)$ for some increasing, concave function $f_v$.  The {\sc ranking} case can be modeled as follows: Voter $v$ has utility $u_{vk}$ for candidate $k \in \C$, and we set $u_v(S) = \max_{k \in S} u_{vk}$. The Chamberlin-Courant rule sets $u_{vk} = m - \ell$ ({\em i.e.}, the Borda score) if candidate $k$ is ranked at position $\ell$ by voter $v$.


\subsection{Group Fairness via Stability}
The notion of fairness we study is defined for arbitrary ordinal preferences over committees. We study fairness via the notion of stability, which has been extensively studied (in similar or more restricted forms) in economics~\cite{lindahlCore,scarfCore,coreConjectureCounter} and computer science~\cite{Fain2016,FainMS18,Brill,PJR2018}. We first define the notion of {\em capture count}.

\begin{define} [Capture Count]
\rm
Given two committees $S_1, S_2 \subseteq \C$, the \emph{capture count} of $S_2$ over $S_1$ is the number of voters who strictly prefer $S_2$ to $S_1$:
\[
V(S_1,S_2) = | \{ v \in \N \ | \ S_2 \succ_v S_1\}| \; . \]
\end{define}

We are now ready to define the notion of stability:
\begin{define}[Stable Committees]
\label{def:stable}
\rm
Given a committee $S \subseteq \C$ of size $K$, we say that a committee $S' \subseteq \C$ blocks $S$ iff
\[ V(S,S') \ge \frac{|S'|}{K} \cdot n \; . \]
For any $1 \le L \le K$, a committee $S$ is $L$-\emph{stable} if there are no committees $S'$ of size {\em at most} $L$ that block it. We term $K$-stable as just ``stable''.
\end{define}

In other words, for any sub-group of voters of size $\beta n$ for some $\beta \in (0,1]$, the committee $S$ is fair in the sense that this sub-group cannot form another committee of size $\lfloor \beta K \rfloor$  so that all members of this sub-group are strictly better off. The notion of $L$-stability restricts this definition to ``small'' subgroups with $\beta \le \frac{L}{K}$.

The classical interpretation of this concept in Economics~\cite{lindahlCore,scarfCore,coreConjectureCounter} is in terms of {\em taxation}: Each voter has a dollar of money. If the goal is to find a committee of size $K$, we assume each committee member has a cost of $\frac{n}{K}$. A committee $S$ of size $K$ is stable if there is no subset of voters who can deviate and ``buy'' another committee $S'$ using their share of money, so that all voters in this deviating coalition are strictly better off with $S'$ than with $S$. Note that the amount of money required to buy $S'$ is precisely $\frac{|S'|}{K} \cdot n$, and the amount of money available with the deviating coalition of voters who strictly prefer $S'$ to $S$ is precisely the capture count $V(S,S')$. Therefore $S'$ is blocking if the deviating coalition of voters have sufficient funds to pay for $S'$.

\medskip
Note that a stable outcome $S$ is (weakly) Pareto-optimal among committees of size at most $K$; if it were not, consider a Pareto-dominating committee $S'$. Since $V(S,S') = n$, the committee $S'$ would be blocking, which contradicts the stability of $S$. In fact, it says something stronger: For every coalition of voters, a stable committee is also Pareto-optimal relative to committees whose size is suitably scaled down. Stability also logically implies other notions of stability considered in multi-winner election literature. For instance, the special case of $1$-stability has been extensively studied as {\em justified representation}, and we discuss these connections further in Section~\ref{sec:related}.

\medskip
In the {\sc Ranking} setting, we can make the following simple observation: $V(S,S') \le \sum_{j \in S'} V(S,j)$, since each deviating voter $v$ prefers $S'$ because of her favorite candidate in it. Therefore, by an averaging argument, if $S'$ blocks $S$, then so does at least one $j \in S'$. This immediately gives the following observation:

\begin{observation}
\label{obs1}
In the {\sc Ranking} setting, a committee  is stable iff it is $1$-stable.
\end{observation}

\subsection{Stable Lotteries and Approximate Stability}
Given the strength of the definition, it is no surprise that stable committees may not exist for the {\sc ranking} setting. Consider the following example with cyclic preferences:

\begin{example}
There are $n = 6$ voters $\{1,2,3,4,5,6\}$ and $m = 6$ candidates $\{a,b,c,d,e,f\}$.
We need to choose a committee of size $K = 3$. The rankings of the voters are as follows:
\[\begin{array}{c|c}
\mbox{Voter} &  \mbox{Preferences} \\
\hline
1 & a \succ b \succ c \succ d \succeq e \succeq f \\
2 & b \succ c \succ a \succ d \succeq e \succeq f \\
3 & c \succ a \succ b \succ d \succeq e \succeq f \\
4 & d \succ e \succ f \succ a \succeq b \succeq c \\
5 & e \succ f \succ d \succ a \succeq b \succeq c \\
6 & f \succ d \succ e \succ a \succeq b \succeq c \\
\end{array} \]
Note that any committee chooses either at most one candidate from $\{a,b,c\}$ or at most one candidate from $\{d,e,f\}$. Assume w.l.o.g. that the former happens, and candidate $a$ is chosen (or no one in $\{a, b, c\}$ is chosen). Then voters $\{2,3\}$ can deviate and choose candidate $c$.  Therefore, there is no stable committee on this example.
\end{example}

In light of this impossibility result, we extend Definition~\ref{def:stable} to allow randomization in choosing the committee. Given a committee size $K$, we let $\Delta$ denote a distribution (or lottery) over committees of size $K$. Our first contribution is the following definition of stable lotteries.

\begin{define}[Stable Lotteries]
\label{def:lottery}
\rm
A distribution (or lottery) $\Delta$ over committees of size $K$ is said to be $L$-\emph{stable} ($1 \le L \le K$) iff for all committees $S' \subseteq \C$ with $|S'| \le L$, we have:
\[ \E_{S \sim \Delta} \left[ V(S,S') \right] < \frac{|S'|}{K} \cdot n \; . \]
We term $K$-stable lotteries as just ``stable''.
\end{define}

In the taxation interpretation, the above definition says the following: For any committee $S'$, if in every realization $S$ of $\Delta$, the voters who strictly prefer $S'$ over $S$ pay for $S'$, then in expectation there is not enough money to pay for $S'$.  This justifies why the voters do not deviate to $S'$ given $\Delta$. Further, we note that implementing a lottery is feasible in several computer systems applications. For instance, consider for example caching data items that are shared by several applications discussed above. In this context, a lottery over possible cache allocations can be implemented by splitting time into chunks, and caching the allocations in the lottery in different chunks of time.

\medskip
We note for the {\sc Ranking} setting, Observation~\ref{obs1} extends to lotteries because
\[ \E_{S \sim \Delta} \left[ V(S,S') \right] \le \sum_{j \in S'} \E_{S \sim \Delta} \left[ V(S,j) \right] \; . \]
\begin{observation}
\label{obs2}
In the {\sc Ranking} setting, a lottery is stable iff it is $1$-stable.
\end{observation}

\paragraph{Approximate Stability.} We next define approximately stable lotteries; this notion will play an important role in our computational results.

\begin{define}[$\epsilon$-approximate Stability]
\rm
For any $\epsilon > 0$, a distribution $\Delta$ over committees of size $K$ is said to be $\epsilon$-\emph{approximately} $L$-\emph{stable} if for all committees $S' \subseteq \C$ with $|S'| \le L$, we have:
\[ \E_{S \sim \Delta} \left[ V(S,S') \right] \le (1+\epsilon) \frac{|S'|}{K} \cdot n \; . \]
We term $\epsilon$-approximately $K$-stable lotteries as just $\epsilon$-approximately stable.
\end{define}

\subsection{Our Results}
Our main result in this paper is the following theorem, which states that a stable lottery always exists and it is reasonably tractable.
Recall that (as in Definition~\ref{def:lottery}) a lottery is $L$-stable if there are no committees of size {\em at most} $L$ that block it, and stable lotteries refer to $K$-stable lotteries.

\begin{theorem}[Main]
\label{thm:main}
In both the {\sc approval set} and {\sc ranking} settings with $m$ candidates, for any committee size $K$, a stable lottery over committees of size $K$ always exists. Furthermore, for any $\epsilon > 0$, an $\epsilon$-approximately $L$-stable lottery can be computed in $\mbox{poly}\left(m^L, \frac{1}{\epsilon} \right)$ time.
\end{theorem}

Note that when combined with Observation~\ref{obs1}, this implies a running time of $\mbox{poly}\left(m,\frac{1}{\epsilon}\right)$ 
to find an $\epsilon$-approximately stable lottery in the {\sc Ranking} setting. Interestingly, assuming we can ask voters to compare any two committees, the algorithm for the {\sc Approval set} and {\sc Ranking} settings are exactly the same, and as the number of voters becomes large, the average number of queries asked to any individual voter goes to zero. Since the algorithm is the same in both cases, as a simple corollary, it also shows the existence of a stable solution when there is a mix of voters in the population, some with {\sc Approval set} preferences and others with {\sc Ranking} preferences.

We note that Theorem~\ref{thm:main} does not follow by analyzing extant voting rules. For instance, in the {\sc approval set} setting, Aziz {\em et al.}~\citet{Brill} showed that the voting rule,  Proportional Approval Voting (PAV)~\cite{Thiele,PJR2018}, that satisfies justified representation, fails to find a stable outcome. In PAV, if voter $v$ approves $r_v$ candidates in the committee,  this voter is assigned a score $s_v = 1 + \frac{1}{2} + \cdots + \frac{1}{r_v}$; the winning committee maximizes $\sum_v s_v$.   We strengthen their example to show that PAV cannot find better than an $O(\sqrt{K})$-approximately stable committee (Theorem~\ref{lem:PAV_lower} in Section~\ref{sec:det}). Further, as mentioned above, deterministic stability is simply not possible for the {\sc ranking} setting, which rules out trying to prove the above result via analyzing any deterministic voting rule. In that sense, we find the existence of stable lotteries quite surprising.

Though not the main focus of the paper, we finally consider the existence of deterministic stable committees in the {\sc approval set} setting. We make progress on this question and show that when the committee size is $K \le 3$, a stable committee always exists (Theorem~\ref{thm:three} in Section~\ref{sec:det}). Prior to this, the only results known (even for $K = 3$) were the existence of committees that are approximately stable~\cite{FainMS18}.

\paragraph{Techniques.} The most interesting aspect of Theorem~\ref{thm:main} is its proof, which uses the probabilistic method. We show that the question of whether stable lotteries exist reduces to deciding if a zero-sum game has negative value. The dual problem involves finding a stable solution given a lottery over blocking committees. We show its existence by developing a {\em rounding procedure} for the dual problem. This procedure performs {\em probability matching}, and simply chooses candidates with probability proportional to their marginal probability in the blocking lottery. We then  argue that this rounding procedure always has negative expected value. Showing this for the {\sc Approval set} setting requires proving a new deviation inequality (Lemma~\ref{lem:pm}) for sums of random variables.


\begin{lemma}[Probability Matching Lemma]
\label{lem:pm}
Let $X$ be the sum of independent Bernoulli random variables: $X = X_1 + X_2 + \cdots + X_n$ where $X_i \sim$ Bernoulli$(1,p_i)$. Let $Y$ be any non-negative \textbf{integer} random variable with $\E[Y] \le \beta\E[X]$. Then,
\[
\Pr[X < Y] < \beta \; .
\]
\end{lemma}

Since we need each committee in the lottery to have size $K$, we combine this large deviation inequality with dependent rounding techniques~\cite{GandhiKS01} to preserve committee size. This rounding procedure can be plugged in as an oracle to the multiplicative weight update method~\cite{AroraHK}, yielding the desired running time for the computational problem of finding an $L$-stable lottery.   We believe the template of our proofs via large deviation inequalities and probabilistic method may find further applications in the emerging theory of fair allocations.

\paragraph{RoadMap.} We present the proof of Theorem~\ref{thm:main} in Section~\ref{sec:exist}. We split this proof into a generic rounding portion (Sections~\ref{sec:zero} and~\ref{sec:round}) and the subsequent part that is specific to {\sc Approval set} (Section~\ref{sec:approval}) and {\sc Ranking} (Section~\ref{sec:ranking}). We present an efficient implementation in Section~\ref{sec:efficient}. In Section~\ref{sec:det}, we consider the notion of deterministic stability. We show the existence of stable committees for $K = 3$, and show that the PAV rule cannot be better than $O(\sqrt{K})$-approximately stable.  We 
conclude with open questions in Section~\ref{sec:open}.

%% file: related.tex

\subsection{Related Work}
\label{sec:related}
We now compare our notion of stability with other notions of stability and fairness extant in literature. We first note that there is a rich literature on using lotteries to achieve fairness, including the influential papers of Hylland and Zeckhauser~\citet{HyllandZ} and Bogomolnaia and Moulin~\citet{Moulin}. Our work is more in the spirit of the latter in the sense that our model does not require voters to specify cardinal utilities over committees.

\paragraph{The Lindahl Equilibrium.} Our notion of stability is inspired by the notion of \emph{core} in cooperative game theory and was first phrased in game theoretic terms by Scarf~\citet{scarfCore}. It has been extensively studied in public goods settings~\cite{lindahlCore,coreConjectureCounter,Fain2016}. Much of this literature considers convex preferences, which translates to voters having preferences over lotteries, and deviating to another lottery if their expected utility increases. In other words, for additive utility functions $u_v$, we say a lottery  $\Delta'$ blocks $\Delta$ if there is a coalition of voters of size at least $n \cdot |S'|/K$ such that for all voters $v$ in this coalition, we have:
\[ \E_{S' \sim \Delta'} \left[u_v(S') \right] \ge \E_{S \sim \Delta} \left[u_v(S) \right] \]
with at least one inequality strict. This notion builds on the seminal work of Foley~\citet{lindahlCore} on the Lindahl market equilibrium. In this equilibrium, each candidate is assigned a per-voter price. If the voters choose their utility maximizing allocation subject to spending a dollar, then (1) they all choose the same outcome; and (2) for each chosen candidate, the total money collected pays for that candidate. It can be shown via a fixed point argument that such an equilibrium pricing always exists when lotteries are allowed, and is a core outcome. Though this existence result is very general, holding for any compact convex set of preferences, it is not known how to compute such a core outcome efficiently even for the committee selection problem with ranking or approval set utilities. 

Our notion of stability coincides with the notion of core if no randomization is allowed, but our notion of stable lotteries differs from how randomization is used in the core. In a Lindahl equilibrium, the voter compares the expected utility from the lottery with the utility on deviation, while in our notion, the lottery is first realized and subsequently the voters who see higher utility will deviate. In a sense, our notion justifies to any coalition of candidates that they do not have enough support given the current lottery, while the Lindahl equilibrium justifies to each coalition of voters that their utility is Pareto-optimal. This difference is subtle, but makes the two notions incomparable -- the existence of one type of stable solution does not imply the other in any obvious way. 

Our approach has two key advantages. First, our notion only requires voters to specify ordinal preferences over committees, and not over lotteries. Note that for both {\sc approval set} and {\sc ranking} settings, preferences over committees do not automatically imply preferences over lotteries -- the latter must be explicitly specified via the choice of utility function. Secondly, our existence theorems, though specific to committee selection, have constructive proofs that lead to efficient algorithms (while fixed point arguments in general do not). For instance, we can compute a stable lottery for {\sc Ranking} in polynomial time, while we do not know how to compute a Lindahl equilibrium efficiently in this setting.

\paragraph{Proportional Fairness.} A different notion of stability is the following: A deviating coalition of voters gets to choose a committee of size $K$, but its utility is scaled down proportional to the size of the coalition. In other words, given a coalition of voters of size $\beta n$, a committee $S'$ of size $K$ is blocking to a lottery $\Delta$ if for all voters $v$ in the coalition, 
\[ \beta u_v(S') \ge \E_{S \sim \Delta} \left[u_v(S) \right] \]
with at least one inequality strict.  In this notion, it is relatively easy to show that maximizing the Nash product of voter utilities finds a stable lottery~\cite{Fain2016}; furthermore, the discrete analog of Nash welfare, proportional approval voting (PAV), finds an approximately stable solution for additive utilities~\cite{FainMS18}. Though this approach is computationally tractable, it explicitly needs voters to have cardinal utilities, and is otherwise incomparable to scaling down the committee size. Our notion is closer to how fairness has been thought about in the Economics literature on public goods, and to how it has been thought about in the multi-winner election literature. We have already delved into the former; we consider the latter next.

\paragraph{Justified Representation.} Deterministic stability in the {\sc approval set} setting logically implies a number of fairness notions considered in multi-winner election literature with approval set preferences, such as justified representation, extended justified representation~\cite{Brill}, and proportional justified representation~\cite{Sanchez,PJR2018}.  The idea behind all these proportional representation axioms is to define a notion of {\em cohesive} groups of agents that all approve a small set of candidates, and ensure that such groups of voters do not deviate, {\em i.e.}, that these groups are proportionally represented in the outcome. As mentioned above, justified representation is exactly the same as our $1$-stability. Similarly, in extended justified representation (EJR),  we only consider deviations by sub-populations of voters of size at least $\alpha n$, if they all approve the {\em same set} of $\lfloor \alpha K \rfloor$ candidates. Therefore, a committee $S$ satisfies EJR iff some voter in this sub-population approves at least $\lfloor \alpha K \rfloor$ in $S$. 

Unlike justified representation and its generalizations, stability is a general condition that holds for all coalitions of agents, not just those that are cohesive. We do pay a price for this generalization: our stability results are for lotteries of committees, and unlike justified representation and EJR, our algorithm for computing stable lotteries runs in polynomial time only when the deviating committee has constant size.  As noted before, it is an open question whether (deterministic) stable committees exist in the {\sc approval set} setting.


\paragraph{Committee Scoring Rules.} Elkind {\em et al.}~\citet{ElkindFSS17}  show that given preferences of voters over candidates, a large class of voting rules for committee selection can be expressed as follows: For each voter, sort the positions of the committee members in its own ranking in non-decreasing order. Now apply a monotone scoring function to this vector of positions, and add up this score for all voters. The winning committee maximizes this score. This can be interpreted as assigning a utility function for each voter for each committee, and finding the committee that maximizes the sum of utilities over all voters, {\em i.e.}, the social optimum. In the {\sc ranking} setting, the voter derives utility only from the most preferred candidate in the committee; in fact, the Chamberlin-Courant voting rule~\citet{Chamberlain}  sums over all the voters, the Borda score of that voter's highest ranked candidate in the committee, and chooses the committee with the highest score. Other voting rules such as $k$-Borda assume more general submodular utilities, and it is an open question whether stable lotteries exist when voters' utilities are submodular. 

Elkind {\em et al.}~\citet{ElkindFSS17}  also show that for several natural scoring functions including Chamberlin-Courant, the resulting voting rule is  {\sc NP-Hard}.  Further, since these voting rules are deterministic, none of them can produce a stable outcome.  Our contribution is to show that for the {\sc ranking} setting,  stable {\em lotteries} always exists and can be computed in polynomial time. 


%% file: approval-proof.tex

\renewcommand{\S}{\mathcal{S}}
\section{Existence of Stable Lotteries}
\label{sec:exist}
In this section, we will prove Theorem~\ref{thm:main}. We will do this in a sequence of steps. First, we will formulate stability as a zero-sum game, and take its dual. Next, we will devise a rounding procedure to show the dual has negative value, which will imply the existence of a stable solution. Though the rounding procedure is the same for both {\sc Approval set} and {\sc ranking} settings, we need different proofs of existence for each. In Section~\ref{sec:efficient}, we show that the rounding procedure also gives an efficient algorithm for finding an $\epsilon$-approximately stable lottery via the multiplicative weight method.

\subsection{Dual Formulation of Stability}
\label{sec:zero}
Recall the definition of a stable lottery from Definition~\ref{def:lottery}. The existence of a stable lottery is equivalent to showing that:
\begin{equation}
\label{eq:main}
 \min_{\Delta} \max_{S'} \left( \E_{S \sim \Delta} \left[ V(S, S') \right] - n \frac{|S'|}{K} \right) < 0
 \end{equation}
where $\Delta$ is a lottery over committees of size $K$.

We formulate a zero-sum game with an attacker and defender. The goal of the defender is to choose a stable committee, while the goal of the attacker is to find a blocking committee. The defender's pure strategies are committees $S_d$ of size {\em exactly} $K$, and the attacker's pure strategies are committees $S_a$ of size {\em at most} $K$. Denote the former set of strategies as $\S_d$ and the latter by $\S_a$.  Given a pair of pure strategies $(S_d,S_a)$ the payoff to the attacker is:
\[ Q(S_d, S_a) = V(S_d, S_a) - n \frac{|S_a|}{K} \; . \]
and the defender tries to minimize this value.

Therefore, the LHS of Equation (\ref{eq:main}) is the value of the game if the defender goes first and chooses mixed strategy $\Delta_d$, and the attacker subsequently chooses strategy $S_a$ to maximize the value of the game. Therefore,  the existence of a stable lottery is equivalent to asking whether the value of this zero-sum game is negative.

\paragraph{Dual Formulation.} Suppose the attacker goes first and chooses a lottery $\Delta_a$ over $\S_a$. Subsequently, the defender chooses a committee $S_d \in \S_d$ to minimize the value of the game. Using duality in zero-sum games,
\[ \min_{\Delta_d} \ \max_{S_a \in \S_a} \E_{S_d \sim \Delta_d} \left[Q(S_d, S_a) \right] = \max_{\Delta_a} \ \min_{S_d \in \S_d} \E_{S_a \sim \Delta_a} \left[Q(S_d, S_a) \right] \; . \]

Define $\beta = \frac{\E_{S_a \sim \Delta_a} [|S_a|]}{K}$.
Note that $\beta \in (0, 1]$ because $\Delta_a$ is a lottery over $\S_a$, {\em i.e.}, over committees of size at most $K$.
Then, we have:
\[ \E_{S_a \sim \Delta_a} \left[ V(S_d, S_a) - n \frac{|S_a|}{K}  \right] = \sum_{v \in \N} \left( \Pr_{S_a \sim \Delta_a} \left[ S_a \succ_v S_d \right] \right) - \beta n = \sum_{v \in \N} \left( \Pr_{S_a \sim \Delta_a} \left[ S_a \succ_v S_d \right]  - \beta \right) \; . \]

Therefore, the value of the game is:
\begin{equation}
\label{eq:main2}
\max_{\Delta_a} \ \min_{S_d \in \S_d} \E_{S_a \sim \Delta_a} \left[Q(S_d, S_a) \right] =  \max_{\Delta_a} \ \min_{S_d \in \S_d} \left( \sum_{v \in \N} \left( \Pr_{S_a \sim \Delta_a} \left[ S_a \succ_v S_d \right]  - \beta \right) \right) \; .
\end{equation}

In order to prove the existence part of Theorem~\ref{thm:main}, we need to show that the RHS of the above identity is negative for the {\sc approval set} and {\sc ranking} settings.

\subsection{Defender's Strategy: Probability Matching}
\label{sec:round}
Fix any mixed strategy $\Delta_a$ for the attacker. Let $\beta = \frac{\E_{S_a \sim \Delta_a} [|S_a|]}{K}$. We will construct a mixed defending strategy $\Delta_d$ over $\S_d$, {\em i.e.} over committees of size $K$, so that
\begin{equation}
\label{eq:main3}
 \forall v \in \N \qquad \Pr_{S_a \sim \Delta_a, S_d \sim \Delta_d} \left[ S_a \succ_v S_d \right] < \beta \; .
\end{equation}
Equations (\ref{eq:main2}) and (\ref{eq:main3}) together show the existence of a defending strategy $S_d$ that makes the attacker's payoff negative:
\[
 \min_{S_d \in \S_d} \E_{S_a \sim \Delta_a} \left[Q(S_d, S_a) \right]
 \le \E_{S_d \sim \Delta_d} \E_{S_a \sim \Delta_a} \left[Q(S_d, S_a) \right]
 = \sum_{v \in \N} \left( \Pr_{S_a \sim \Delta_a, S_d \sim \Delta_d} \left[ S_a \succ_v S_d \right]  - \beta \right) < 0 \; .
\]
Because such $S_d$ exists for any mixed attacking strategy $\Delta_a$, the value of the game is negative, and hence a stable lottery always exists.

\paragraph{Probability Matching.}
We will construct the mixed defending strategy using {\em probability matching}. For candidate $i \in \C$, let random variable $Y_i = \mathbf{1}_{i \in S_a}$ be the indicator variable for the event $i \in S_a$ when $S_a \sim \Delta_a$, and let $p_i = \E[Y_i] = \Pr_{S_a \sim \Delta_a}[i \in S_a]$. Note that the inclusion of different candidates in the attacking strategy can be correlated, {\em i.e.}, the $Y_i$'s can be correlated.

Let $q_i = \min(1,p_i/\beta)$.  The defender's strategy will include $i \in S_d$ with probability at least $q_i$, {\em i.e.}, the defender will {\em probability match} the attacker for each candidate $i \in \C$. We first make a simple observation.

\begin{claim}
$\sum_{i \in \C} q_i \le K$.
\end{claim}
\begin{proof}
The expected size of the attacker committee is $\beta K$, so $\sum_i p_i = \beta K$. Therefore,
\[ \sum_i q_i = \sum_i \min(1,p_i/\beta) \le \sum_i p_i/\beta = K \; . \qedhere \]
\end{proof}

\paragraph{Randomized Dependent Rounding.}
%
We first increase $\vec{q}$ arbitrarily to $\vec{\alpha}$ so that $\sum_i \alpha_i = K$ and $\alpha_i \in [q_i,1]$ for all $i \in \C$.
We will construct a mixed defending strategy $\Delta_d$ over committees $S_d$ of size $K$ by performing {\em randomized dependent rounding}~\cite{GandhiKS01} on $\vec{\alpha}$.
Let random variable $X_i$ be the indicator variable for the event $i \in S_d$ when $S_d \sim \Delta_d$.
We omit the details of this rounding procedure, and instead summarize its properties in the following lemma:

\begin{lemma} [Gandhi {\em et al.}~\cite{GandhiKS01}]
\label{lem:depend}
The distribution $\Delta_d$ obtained by performing randomized dependent rounding on $\vec{\alpha}$ satisfies the following:
\begin{enumerate}
\item $\Delta_d$ is a distribution over committees of size $K$, {\em i.e.}, for every $S_d \sim \Delta_d$ we have $|S_d| = K$.
\item $\Pr[X_i = 1] = \alpha_i$, {\em i.e.}, the marginal probability of choosing candidate $i$ is preserved.
\item The random variables $X_i$ are negatively correlated.  More specifically, for every subset $S \subseteq \C$,
\[ \Pr\left[\bigwedge_{i \in S} X_i = b\right] \le \prod_{i \in S} \Pr[X_i = b] \quad \forall b \in \{0, 1\} \; . \]
\end{enumerate}
\end{lemma}

This completes the generic portion of the proof that there exists a stable lottery over committees.
It remains to show that the defender's lottery $\Delta_d$ we constructed satisfies Inequality~(\ref{eq:main3}).
We will show this separately for the {\sc Approval set} and {\sc ranking} settings.

\subsection{Proof of Theorem~\ref{thm:main} for {\sc Approval Set} Setting}
\label{sec:approval}
We will now show Inequality~(\ref{eq:main3}) holds for the {\sc Approval set} setting.  As we show later, this statement is a stronger version of the Probability Matching Lemma (Lemma~\ref{lem:pm}), where we replace independence of $\{X_i\}$ by negative dependence.

Fix any voter $v \in \N$ with approval set $A_v$. By the definition of {\sc Approval Set}:
\[ \Pr_{S_a \sim \Delta_a, S_d \sim \Delta_d} \left[ S_a \succ_v S_d \right] = \Pr_{S_a \sim \Delta_a, S_d \sim \Delta_d} \left[ |S_a \cap A_v| > |S_d \cap A_v| \right] \; . \]
Suppose $A_v = \{c_1, c_2, \ldots, c_{\ell}, c_{\ell+1}, \ldots, c_r\}$. We simplify notation to denote $\alpha_i = \alpha_{c_i}$. These candidates are ordered so that $\alpha_i$'s are in ascending order. Further, $\alpha_{\ell} < 1$ and $\alpha_{\ell+1} = \alpha_{\ell+2} = \cdots = \alpha_r = 1$. Since $c_{\ell+1}, \ldots, c_r$ are selected in $S_d$ with probability $1$,
\[
\Pr\left[|S_d \cap A_v| < |S_a \cap A_v|\right] \leq \Pr\left[|S_d \cap \{c_1, \ldots, c_{\ell}\}| < |S_a \cap \{c_1, \ldots, c_{\ell}\}|\right] \; .
\]
Defining $T_d := S_d \cap \{c_1, \ldots, c_{\ell}\}$ and $T_a := S_a \cap \{c_1, \ldots, c_{\ell}\}$, we have
\begin{equation}
\label{eq:2} \Pr_{S_a \sim \Delta_a, S_d \sim \Delta_d} \left[ S_a \succ_v S_d \right] \le \Pr\left[|T_d| < |T_a|\right] \; .
\end{equation}
In order to show Inequality~(\ref{eq:main3}), our goal is therefore to show that $\Pr\left[|T_d| < |T_a|\right] < \beta$.

\medskip


Recall that $Y_i$ is the indicator random variable for whether $c_i \in S_a$ when $S_a \sim \Delta_a$, and similarly, $X_i$ is the indicator random variable for whether $c_i \in S_d$ when $S_d \sim \Delta_d$.
The $Y_i$'s can be correlated and $E[Y_i] = p_i$.
For $i \in \{1,2,\ldots, \ell\}$, by Lemma~\ref{lem:depend}, we have $\E[X_i] = \alpha_i \ge p_i/\beta$ and the $X_i$'s are negatively correlated.
Define random variables $X = \sum_{i=1}^\ell X_i$ and $Y = \sum_{i=1}^\ell Y_i$.

Using these notations, our goal is to prove
\[ \Pr\left[|T_d| < |T_a|\right] = \Pr \left[ X < Y \right] < \beta \; . \]

The next lemma uses Chernoff bound to upper bound the probability that $X < \eta$.

\begin{lemma}
\label{lem:deviation-chernoff}
Let $X = \sum_{i=1}^\ell X_i$ and $\mu = \E[X]$.
For any integer $\eta \ge 1$, $\Pr[X < \eta] < \frac{\eta}{\mu}$.
\end{lemma}
\begin{proof}
The statement holds trivially when $\eta \ge \mu$, so w.l.o.g. we can assume $1 \le \eta < \mu$.
Because $\eta$ is an integer, $\Pr[X < \eta] = \Pr[X \le \eta-1]$.

By Lemma~\ref{lem:depend}, the $X_i$'s produced by randomized dependent rounding \cite{GandhiKS01} is negatively correlated, so we can apply Chernoff bound to the $X_i$'s~\cite{PanconesiS97,GandhiKS01}:
\[ \Pr[X \le \eta-1] = \Pr\left[X \le \left(\frac{\eta-1}{\mu}\right) \mu\right] \le \exp\left(-\frac{1}{2} \left(1-\frac{\eta-1}{\mu}\right)^2 \mu \right) \; .\]

It remains to show that the RHS is less than $\frac{\eta}{\mu}$ for any $\mu > 0$ and any integer $1 \le \eta < \mu$.
One way to do so is to fix the ratio $r = \frac{\eta}{\mu} \in (0, 1)$ and first minimize the RHS over $\mu = \frac{\eta}{r} \ge \frac{1}{r}$.

Let $f(\mu) = (1-r+\frac{1}{\mu})^2 \mu$.
We can compute its derivative: $f'(\mu) = (1-r)^2 - \frac{1}{\mu^2}$.
Thus, $f(\mu)$ is decreasing when $0 < \mu \le \frac{1}{1-r}$ and increasing when $\mu \ge \frac{1}{1-r}$. 

When $r \ge \frac{1}{2}$, we have $\min_{\mu \ge \frac{1}{r}} f(\mu) = f(\frac{1}{1-r}) = 4(1-r)$, which implies $\mathrm{RHS} \le \exp\left(-\frac{4(1-r)}{2}\right) < r$.
When $r \le \frac{1}{2}$, we have $\frac{1}{r} \ge \frac{1}{1-r}$, so $\min_{\mu \ge \frac{1}{r}} f(\mu) = f(\frac{1}{r}) = \frac{1}{r}$ and $\mathrm{RHS} \le \exp\left(-\frac{1}{2r}\right) < r$.
\end{proof}

Given Lemma~\ref{lem:deviation-chernoff}, we can prove Inequality~\eqref{eq:main3} by considering all possible values of $Y$:
\[ \Pr[X < Y]
 = \sum_{\eta=1}^\ell \Pr[Y = \eta] \cdot \Pr[X < \eta]
 < \sum_{\eta=1}^\ell \Pr[Y = \eta] \frac{\eta}{\E[X]}
 = \frac{\E[Y]}{\E[X]} \le \beta \; . \]

The last step follows from the fact that $\E[X] = \sum_{i=1}^\ell \alpha_i \ge \sum_{i=1}^\ell p_i / \beta = \E[Y] / \beta$. This completes the proof that a stable lottery exists in the {\sc Approval Set} setting.

\paragraph{Discussion.} 
The above proof also proves Lemma~\ref{lem:pm}.
Recall that Lemma~\ref{lem:pm} states that
\[ \Pr \left[ X < Y \right] < \beta \;  \]
where $X$ is the sum of independent Bernoulli random variables and $Y$ is a non-negative integer random variable. 
Now since $\E[Y] \leq \beta\E[X]$, a simple application of Markov's inequality yields:
\[ \Pr \left[ \E[X] < Y \right] < \beta \; . \]
However, this is not quite what we want. The quantity $X = \sum_i X_i$  is the sum of independent Bernoulli random variables, and will deviate below its expectation. The proof in Section~\ref{sec:approval} varies the threshold $\eta$, and
uses Chernoff bound to upper bound the probability that $X$ deviates below $\eta$.

The probability matching approach chooses candidates in $S_d$ ({\em i.e.}, the variables $X_i$) proportional to the marginals of the correlated distribution $\Delta_a$ ({\em i.e.}, the variables $Y_i$). This is similar to the randomized contention resolution scheme for approximately maximizing submodular welfare~\cite{FeigeV}. However, this resemblance is superficial. First, we do not round with probability equal to the marginal, but instead need to scale it up by $1/\beta$. More importantly, the analysis for submodular welfare focuses on achieving a constant approximation against an LP bound, {\em i.e.}, they show that when $\beta = 1$, then $\E[f(\vec{X})] \ge (1-1/e) \E[f(\vec{Y})]$ for submodular function $f$. On the other hand, we show a more delicate statement $\Pr \left[ \sum_{i=1}^{\ell} Y_i > \sum_{i=1}^{\ell} X_i  \right] < \beta$. In fact, we can modify the tight example in~\cite{FeigeV} to show that for this specific rounding procedure, the statement $\Pr \left[ f(\vec{Y}) > f(\vec{X})  \right] < \beta$ is {\em false} for arbitrary submodular functions $f$. This implies our probability matching procedure {\em cannot} be used to show the existence of stable lotteries when utilities of voters for committees are submodular. It is an open question to extend this proof to the case where the $X_i \sim $ Bernoulli$(s_i,\alpha_i)$, which corresponds to the voters having arbitrary additive utilities for candidates. 


\subsection{Proof of Theorem~\ref{thm:main} for {\sc Ranking} Setting}
\label{sec:ranking}
We now prove Inequality~(\ref{eq:main3}) for the {\sc Ranking} setting. Fix any voter $v \in \N$. Assume the candidates are ordered so that $c_1 \succ_v c_2 \succ_v \cdots \succ_v c_m$.  Let $X_i$ denote the indicator random variable that the defender's committee $S_d \sim \Delta_d$ includes $c_i$, and let $Y_i$ denote the indicator random variable that the attacker's committee $S_a \sim \Delta_a$ includes $c_i$.

Recall that the $X_i$'s are negatively correlated and $Y_i$'s can be arbitrarily correlated, $\E[Y_i] = p_i$, and $\E[X_i] = \alpha_i \ge \min(1, p_i / \beta)$.
Let $c_\ell$ denote the earliest (highest ranked) candidate in this ordering for whom $\alpha_{\ell} = 1$, {\em i.e.}, $X_{\ell} = 1$ with probability $1$. Note that for $1 \le j < \ell$, we have $\alpha_j < 1$ and thus $\alpha_j \ge p_j / \beta$.

Let $c_j$ be the highest ranked candidate picked by the attacker.
That is, $j$ is a random variable which is the smallest index with $Y_j = 1$.
Then, $S_a \succ_v S_d$ iff $X_1 = X_2 = \cdots = X_j = 0$, {\em i.e.}, the defender fails to pick any candidate that ranks at least as high as $j$.

Because the defender will always pick $c_\ell$, for $S_a \succ_v S_d$ to happen, the attacker's highest ranked candidate has to appear before $c_\ell$ ({\em i.e.}, $j < \ell$).
Therefore,
\begin{eqnarray*}
 \Pr[S_a \succ_v S_d] & = & \sum_{j=1}^{\ell-1} \Pr\left[ Y_1 = Y_2 = \cdots = Y_{j-1} = 0 \mbox{ and } Y_j = 1  \mbox{ and } X_1 = X_2 = \cdots = X_j = 0 \right] \\
 & \le & \sum_{j=1}^{\ell-1}  \Pr\left[ Y_j = 1  \mbox{ and } X_1 = X_2 = \cdots = X_j = 0 \right] \; .
  \end{eqnarray*}
Since $\{X_i\}$ is independent of $\{Y_i\}$, we have:
 \begin{eqnarray*}
 \Pr[S_a \succ_v S_d] & \le &  \sum_{j=1}^{\ell-1} \Pr\left[ Y_j = 1\right] \cdot \Pr\left[ X_1 = X_2 = \cdots = X_j = 0 \right] \\
 & \le & \sum_{j=1}^{\ell-1} \Pr\left[ Y_j = 1\right] \cdot \Pr[X_1 = 0] \Pr[ X_2 = 0] \cdots  \Pr[X_j = 0] \\
 & \le & \sum_{j=1}^{\ell-1} \beta \alpha_j \prod_{i=1}^j (1-\alpha_i) \; .
   \end{eqnarray*}
Here, the second inequality follows from the negative correlation of $\{X_i\}$.
More precisely, this inequality is a special case of the negative correlation condition of Lemma~\ref{lem:depend}.
The final inequality follows since $\E[Y_i] = p_i \le \beta \alpha_i$, and $\E[X_i] = \alpha_i$.

Consider the following stopping process: Consider candidates in order $c_1, c_2, \ldots, c_{\ell}$. We stop at candidate $c_1$ with probability $\alpha_1$; if we do not stop, we stop at $c_2$ independently with probability $\alpha_2$, and so on.
Because $\alpha_j < 1$ for all $1 \le j < \ell$, we stop with probability less than $1$ before we reach $c_{\ell}$:
\[ \sum_{j=1}^{\ell-1} \alpha_j \prod_{i=1}^{j-1} (1-\alpha_i) < 1 \; . \]
Comparing this inequality with the previously obtained bound, we have:
\[ \Pr[S_a \succ_v S_d] \le \sum_{j=1}^{\ell-1} \beta \alpha_j \prod_{i=1}^j (1-\alpha_i) \le \beta \sum_{j=1}^{\ell-1} \alpha_j \prod_{i=1}^{j-1} (1-\alpha_i) < \beta \; . \]

This completes the proof of Inequality (\ref{eq:main3}) for the {\sc Ranking} setting.



%% file: runtime.tex
\subsection{Computing Stable Lotteries Efficiently}
\label{sec:efficient}
In this section, we turn to the problem of computing a stable lottery efficiently, and prove the algorithmic part of Theorem~\ref{thm:main}. We will show a running time of poly$(m^L, 1/\epsilon)$ to compute an $\epsilon$-approximately $L$-stable lottery (see Definition~\ref{def:lottery}).  This yields a polynomial running time for $L$-stable lotteries for constant $L$ in the {\sc Approval set} setting, and for stable lotteries in the {\sc Ranking} setting. Since our algorithm is {\em the same} for both {\sc Approval set} and {\sc Ranking} settings (assuming voters have the ability to compare two committees); it also yields a stable lottery when there is a mix of voters, some with {\sc Approval set} and others with {\sc Ranking} preferences.

First observe that since the size of the strategy sets $\S_d$ and $\S_a$ are $O(m^K)$, and since all we are doing is solving a zero-sum game, there is clearly a $\mbox{poly}(n, m^K)$ time algorithm.   The key observation is that the randomized dependent rounding procedure allows us to compute a defender strategy efficiently, and this allows us to compute $L$-stable lotteries in time that depends on $m^L$ instead of $m^K$. We then combine this with estimating the $V(S,S')$ values by sampling voters and asking them to compare $S$ and $S'$, thereby eliminating the dependence of the running time on $n$.  Since the details follow from fairly standard ideas, we only outline the argument.

Given a lottery $\Delta_a$ over $\S_a$, let {\sc Oracle}$(\Delta_a, \epsilon)$ be a procedure that finds $S_d \in \S_d$ such that 
\[ R_{\epsilon}(S_d, \Delta_a) \equiv \E_{S_a \sim \Delta_a} \left[ V(S_d, S_a) - (1+\epsilon) n \frac{|S_a|}{K}  \right] < 0 \; . \]

\begin{claim}
For $\epsilon \in (0,1/5)$, {\sc Oracle}$(\Delta_a, \epsilon)$ can be implemented in expected time $\mbox{poly}\left(m^L, \frac{1}{\epsilon} \right)$, and is correct with high probability.
\end{claim}
\begin{proof}
First note that given $\Delta_a$, the randomized dependent rounding procedure in~\cite{GandhiKS01} can be implemented in $O(m^2)$ time.  The randomized dependent rounding outputs a distribution $\Delta_d$ that satisfies $\E_{T_d \sim \Delta_d}[R_0(T_d,\Delta_a)] < 0$.
We have:
\[ \E_{T_d \sim \Delta_d} \left[\E_{S_a \sim \Delta_a} \left[ V(T_d, S_a) \right] \right]\le \E_{S_a \sim \Delta_a} \left[n \frac{|S_a|}{K} \right] \; . \]
By Markov's inequality,  we have:
\begin{equation}
\label{eq:ten}
 \Pr_{T_d \sim \Delta_d} \left[ \E_{S_a \sim \Delta_a} \left[ V(T_d, S_a) \right] \ge (1+ \epsilon) \E_{S_a \sim \Delta_a} \left[n \frac{|S_a|}{K} \right] \right] < \frac{1}{1+\epsilon} \; .
 \end{equation}

Given an output $T_d$ of randomized dependent rounding, for each $S_a$ of size at most $L$, we can estimate $V(T_d,S_a)$  to within an additive $\frac{\epsilon n}{K}$ as $\hat{V}(T_d,S_a)$, by sampling $\mbox{poly}(m,\frac{1}{\epsilon})$ voters and asking them to compare $T_d$ with $S_a$. This allows us to estimate $\E_{S_a \sim \Delta_a}[V(T_d, S_a)]$ to within an additive $\frac{\epsilon n}{K}$ with high probability in $\mbox{poly}(m^L,\frac{1}{\epsilon})$ time, since the support of $\Delta_a$ has size $O(m^L)$. Since the entire algorithm takes poly$(m,1/\epsilon)$ steps, we can therefore assume that with probability at least $1-1/m$, our estimate of any $V(T_d,S_a)$ is accurate for all steps of the algorithm. Call the resulting estimate of $R_{\epsilon}(T_d, \Delta_a)$, where each $V(T_d,S_a)$ is estimated by sampling, as $\hat{R}_{\epsilon}(T_d, \Delta_a)$.

We first show that for $T_d \sim \Delta_d$, the event $\hat{R}_{2\epsilon}(T_d, \Delta_a) < 0$ happens with probability at least $\epsilon/2$.
Inequality (\ref{eq:ten}) implies $R_{\epsilon}(T_d, \Delta_a) < 0$ with probability at least $\frac{\epsilon}{1+\epsilon}$.
When this event happens, we have $R_{2\epsilon}(T_d, \Delta_a) \le R_{\epsilon}(T_d, \Delta_a) - \frac{\epsilon n}{K} < -\frac{\epsilon n}{K}$.
Since our sampled estimates of $\E_{S_a \sim \Delta_a}[V(T_d, S_a)]$ are accurate within an additive $\frac{\epsilon n}{K}$ with high probability,
  we have $\hat{R}_{2\epsilon}(T_d, \Delta_a) \le R_{2\epsilon}(T_d, \Delta_a) + \frac{\epsilon n}{K} < 0$.

Therefore, if the randomized dependent rounding procedure is repeated till $\hat{R}_{2\epsilon}(T_d, \Delta_a) < 0$, it takes $\mbox{poly}\left(\frac{1}{\epsilon}\right)$ trials in expectation.
We conclude the proof by noting that the resulting $T_d$ is a feasible solution for {\sc Oracle}$(\Delta_a, 3\epsilon)$ with high probability: $R_{3\epsilon}(T_d, \Delta_a) \le R_{2\epsilon}(T_d, \Delta_a) - \frac{\epsilon n}{K} \le \hat{R}_{2\epsilon}(T_d, \Delta_a) < 0$.
\end{proof}

We now use the multiplicative weight update (MWU) method~\cite{AroraHK} in a standard fashion: 

\begin{enumerate}
\item Given the adversary's strategy $\Delta_a^t$ at time $t$, run {\sc Oracle}$(\Delta^t_a, \epsilon)$ to compute $S^t_d$ satisfying $R_{\epsilon}(S^t_d, \Delta^t_a) < 0$. 
\item  Treat each committee $S_a$ of size at most $L$  as an expert, and set its gain to be $g^t(S_a) = V(S^t_d,S_a) - (1+\epsilon) n \frac{|S_a|}{K}$. Again, these gains can be approximately computed by sampling the voters and asking them to compare $S^t_d$ and $S_a$. 
\item Feed these gains to the MWU algorithm, which outputs a lottery $\Delta^{t+1}_a$ over the experts.
\end{enumerate}
 
 By a standard analysis, if we run the procedure for $T = \mbox{poly}\left(m,\frac{1}{\epsilon}\right)$ steps, the lottery $\Delta_d$ that chooses $S_d^t$ for $t = 1,2,\ldots,T$ with equal probability will be an $\epsilon'$-approximately $L$-stable lottery for $\epsilon' = O(\epsilon)$. Further, the algorithm only involves asking voters to compare two committees, and is hence the same for both {\sc Approval set} and {\sc Ranking}. 

%% file: efficient.tex
\section{Deterministic Stability for {\sc Approval set} Setting}
\label{sec:det}
So far, we have considered existence of stable lotteries. As mentioned earlier, for the {\sc Approval set} setting, the existence of stable (deterministic) committees is a tantalizing open question.  In this section, we present some results that make progress towards the goal of finding stable committees. First, we show that stable committees always exist when the committee size is $K \le 3$, regardless of the number of candidates and voters. In addition, we show that the PAV rule that satisfies justified representation, fails to find {\em any non-trivial} approximation to stable committees. This strengthens the negative result in~\cite{Brill} to include inapproximability.

\subsection{Existence of Stable Committees When $K = 3$}
The next theorem states that for {\sc approval set} setting and $K = 3$, stable committee always exists and can be computed efficiently. Recall that $m$ is the number of candidates, $n$ is the number of voters, and $V(S, S')$ is the number of voters who prefer committee $S'$ over $S$.

\begin{theorem}
\label{thm:three}
In the {\sc approval set} setting, for committee size $K = 3$, a stable committee always exists. Moreover, it can be computed in time $O(m^3 n)$.
\end{theorem}
\begin{proof}
For a committee $S$, we use $n_i(S)$ for the number of voters who approve exactly $i$ candidates in $S$,
  and we use $n_{\ge i}(S)$ for the number of voters who approve at least $i$ candidates in $S$.

First observe that it is sufficient to find committee $T$ with no blocking committees of size $1$ or $2$.
This is because either $T$ is stable, or it has blocking committees of size $3$ that Pareto-dominate $T$.
Consider any Pareto-optimal committee $T'$ that Pareto-dominates $T$.
$T'$ is stable because it has no blocking committees of size $3$ due to its Pareto-optimality; and since $T'$ makes all voters happier than they were under $T$, $T'$ has no blocking committees of size $1$ or $2$ either.

The rest of the proof shows that such committee $T$ always exists via the following case analysis:
\begin{enumerate}
\item There exists a committee $S = \{a, b\}$ such that $n_2(S) > \frac{n}{3}$.
\item There are no committees of size $2$ that satisfy (1), but there exists a committee $S = \{a, b\}$ such that $n_{\ge 1}(S) \ge \frac{2n}{3}$.
\item There are no committees of size $2$ that satisfy (1) or (2).
\end{enumerate}

For Case (1), assume w.l.o.g. that $S$ has a blocking committee of size $1$ or $2$.
(Otherwise we can choose $T = S$.)
However, any committee $S'$ of size $2$ is not blocking, because $v(S, S') \le n - n_2(S) < \frac{2n}{3}$.
Therefore, it must be the case that there is a blocking committee $S' = \{c\}$ of size one.
We add $c$ to $S$ and consider the committee $T = \{a, b, c\}$.
Note that $n_{\ge 2}(T) \ge n_2(S) > \frac{n}{3}$, and $n_1(T) \ge v(S, \{c\}) \ge \frac{n}{3}$ because voters who prefer $\{c\}$ over $S$ must approve $c$ but not anyone in $S = \{a, b\}$.
We argue that $T$ satisfies our requirements, because for any committee $T'$ of size $2$, we have $V(T, T') \le n - n_{\ge 2}(T) < \frac{2n}{3}$; and for any committee $T'$ of size $1$, we have $V(T, T') \le n_0(T) =  n - n_{\ge 2}(T) - n_1(T) < \frac{n}{3}$.

For Case (2), assume w.l.o.g. there is a candidate $c$ with $v(S, \{c\}) > 0$.
Consider the committee $T = \{a, b, c\}$.
Note that $n_{\ge 1}(T) = n_{\ge 1}(S) + V(S, \{c\}) > \frac{2n}{3}$.
We argue that $T$ satisfies our requirements,
  because for any $T'$ of size $1$, we have $v(T, T') \le n_0(T) = n - n_{\ge 1}(T) < \frac{n}{3}$;
  and for any $T'$ of size $2$, we have $V(T, T') \le n_2(T') + n_0(T) < \frac{2n}{3}$.
The last inequality is because voters who prefer $T'$ over $T$ must either approve $2$ candidates in $T'$, or approve $0$ candidates in $T$; and $n_2(T') \le \frac{n}{3}$ since we are not in Case (1).

Finally, in Case (3), there are no blocking committee $T'$ of size $2$ because $V(\varnothing, T') = n_{\ge 1}(T') < \frac{2n}{3}$.
This allows us to focus only on blocking committees of size $1$.
There is w.l.o.g. some candidate $a$ such that $V(\varnothing, \{a\}) \ge \frac{n}{3}$, otherwise we can set $T = \varnothing$.
Then again, there is w.l.o.g. some candidate $b$ such that $V(\{a\}, \{b\}) \ge \frac{n}{3}$, otherwise we can set $T = \{a\}$.
However, this contradicts the assumption that we are in Case (3), because $n_{\ge 1}(\{a, b\}) = V(\varnothing, \{a\}) + V(\{a\}, \{b\}) \ge \frac{2n}{3}$.

It takes time $O(m^2 n)$ to find a committee $T$ with no size $1$ and $2$ blocking committees.
This can be done by enumerating all committees of size $2$ in time $O(m^2 n)$ to decide which case we are in, and then the bottleneck is Case (1) in which we can find a blocking committee of size $1$ in time $O(m n)$.
Starting from $T$, we can find a stable solution $T'$ in time $O(m^3 n)$ by enumerating all committees of size $3$ and maintaining the current Pareto-optimal solution.
Hence, the overall running time is $O(m^2 n + mn + m^3 n) = O(m^3 n)$.
\end{proof}

%% file: counterexample.tex
\subsection{Proportional Approval Voting Can Be Far From Stable}
\label{sec:counter}
Recall that the PAV rule works as follows: For any committee $S$, suppose voter $v$ approves $r$ candidates in $S$, then we set the score $q_v(S) = 1 + \frac{1}{2} + \cdots + \frac{1}{r}$. The PAV rule finds $S$ that maximizes $\sum_v q_v(S)$. We show that the PAV rule cannot be better than $O(\sqrt{K})$-approximately stable.

\begin{theorem}
In the {\sc approval set} setting, the PAV rule may output a committee that is not $o(\sqrt{K})$-approximately stable.
\label{lem:PAV_lower}
\end{theorem}
\begin{proof}
Consider the following example (illustrated in Figure~\ref{fig:PAV} where each rectangle is a candidate whose projection on $x$-axis corresponds to the voters who approve her): The set of voters is divided into two equal-sized disjoint sets $\mathcal{N}_L$ and $\mathcal{N}_R$. There are $4$ different sets of candidates:
\begin{enumerate}
\item Set $A$: There are $\frac{P}{2}$ ($P$ is a parameter related to $K$ to be set later) candidates in $A$. The voters who approve a candidate in $A$ are exactly those voters in $\mathcal{N}_L$.
\item Set $B$: There are $\frac{P}{2}$ candidates in $B$. The voters who approve a candidate in $B$ are exactly those voters in $\mathcal{N}_R$.
\item Set $C$: There are $\frac{P}{2} \cdot \frac{P}{4}$ candidates in $C$. Voters in $\mathcal{N}_R$ do not approve candidates in $C$. Every candidate in $C$ is approved by $\frac{\mathcal{N}_L}{P / 4}$ voters. Each voter in $\mathcal{N}_L$ approves $\frac{P}{2}$ candidates in $C$.
\item Set $D$: There are $\frac{P}{2}$ candidates in $D$. Voters in $\mathcal{N}_L$ do not approve candidates in $D$. Every candidate in $D$ is approved by $\frac{\mathcal{N}_R}{P / 2}$ voters. Each voter in $\mathcal{N}_R$ approves one candidate in $D$.
\end{enumerate}

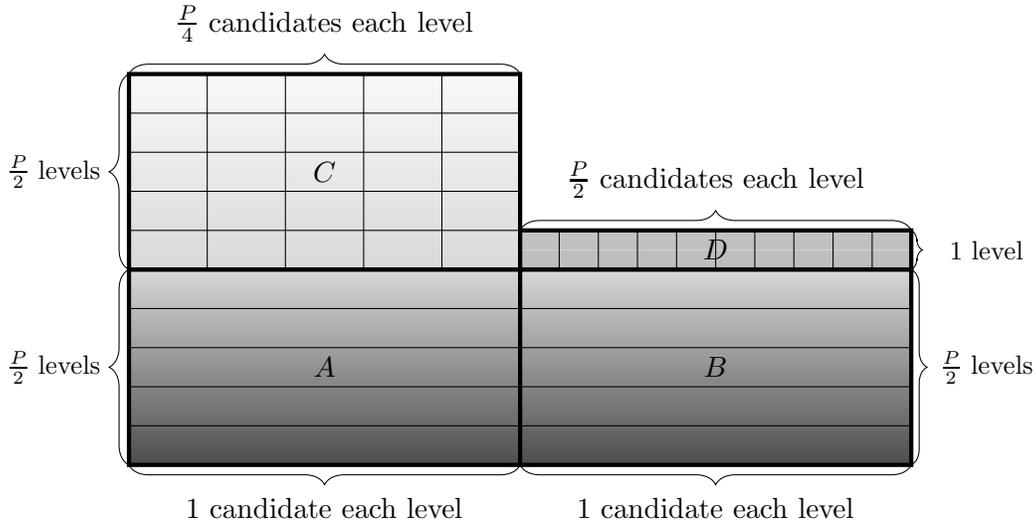
\begin{figure}[htbp]
\centering
\begin{tikzpicture}[scale=0.52]
\filldraw[ultra thick, top color=gray!30!,bottom color=gray!140!] (0,0) rectangle node{$A$} +(10,5);
\filldraw[ultra thick, top color=gray!30!,bottom color=gray!140!] (10,0) rectangle node{$B$} +(10,5);
\filldraw[ultra thick, top color=gray!5!,bottom color=gray!30!] (0,5) rectangle node{$C$} +(10,5);
\filldraw[ultra thick, top color=gray!25!,bottom color=gray!30!] (10,5) rectangle node{$D$} +(10,1);
\draw [decorate,decoration={brace,amplitude=7.5pt},yshift=0pt]
(0,0) -- (0,5.0) node [black,midway,xshift=-1cm] {\small $\frac{P}{2}$ levels};
\draw [decorate,decoration={brace,amplitude=7.5pt},yshift=0pt]
(0,5.0) -- (0,10.0) node [black,midway,xshift=-1cm] {\small $\frac{P}{2}$ levels};
\draw [decorate,decoration={brace,amplitude=6pt,mirror},yshift=0pt]
(20,0) -- (20,5) node [black,midway,xshift=1cm] {\small $\frac{P}{2}$ levels};
\draw [decorate,decoration={brace,amplitude=4.2pt,mirror},yshift=0pt]
(20,5) -- (20,6) node [black,midway,xshift=1cm] {\small $1$ level};
\draw [decorate, decoration={brace, amplitude=7.5pt}] (0,10) -- (10,10) node [midway, above=0.3cm] {$\frac{P}{4}$ candidates each level};
\draw [decorate, decoration={brace, amplitude=7.5pt}] (10,6) -- (20,6) node [midway, above=0.3cm] {$\frac{P}{2}$ candidates each level};
\draw [decorate, decoration={brace, amplitude=7.5pt}] (10,0) -- (0,0) node [midway, below=0.3cm] {$1$ candidate each level};
\draw [decorate, decoration={brace, amplitude=7.5pt}] (20,0) -- (10,0) node [midway, below=0.3cm] {$1$ candidate each level};

\draw (0, 1) -- (20, 1);
\draw (0, 2) -- (20, 2);
\draw (0, 3) -- (20, 3);
\draw (0, 4) -- (20, 4);

\draw (0, 6) -- (10, 6);
\draw (0, 7) -- (10, 7);
\draw (0, 8) -- (10, 8);
\draw (0, 9) -- (10, 9);

\draw (2, 5) -- (2, 10);
\draw (4, 5) -- (4, 10);
\draw (6, 5) -- (6, 10);
\draw (8, 5) -- (8, 10);

\draw (11, 5) -- (11, 6);
\draw (12, 5) -- (12, 6);
\draw (13, 5) -- (13, 6);
\draw (14, 5) -- (14, 6);
\draw (15, 5) -- (15, 6);
\draw (16, 5) -- (16, 6);
\draw (17, 5) -- (17, 6);
\draw (18, 5) -- (18, 6);
\draw (19, 5) -- (19, 6);
\end{tikzpicture}
\caption{PAV cannot be better than $O(\sqrt{K})$-approximately stable.}
\label{fig:PAV}
\end{figure}

Consider running the PAV rule on this instance when $K = P + \frac{P^2}{8}$.
PAV will first choose all candidates in $A$ and $B$.
Note that at this point, any candidate in $C$ has marginal contribution at least $\frac{1}{P} \cdot \frac{\mathcal{N}_L}{P / 4}$, while any candidate in $D$ has marginal contribution at most $\frac{1}{(P/2) + 1} \cdot \frac{\mathcal{N}_R}{P / 2}$, which is strictly smaller.
Therefore, PAV will select the committee $S = A \cup B \cup C$.

However, voters in $\mathcal{N}_R$ can form a coalition and deviate to committee $S' = B \cup D$.
All voters in $\mathcal{N}_R$ are better off, thus $V(S, S') = |\mathcal{N}_R| = \frac{n}{2}$.
For PAV to be $\epsilon$-approximately stable, we need
\[
\frac{n}{2} \leq (1 + \epsilon) \frac{|S'|}{K} \cdot n = (1 + \epsilon) \frac{P}{K} \cdot n \; .
\]
Since $K = P + \frac{P^2}{8}$, we need $\epsilon = \Omega(\sqrt{K})$.
\end{proof}

%% file: conclusions.tex
\section{Conclusions}
\label{sec:open}
We view our results as a first step in understanding the general notion of stability for committee selection problems. We conclude with several open questions. The first major open question is the existence of deterministic stable committees in the {\sc Approval set} setting, generalizing our positive result for $K = 3$ to general $K$. We conjecture that such a stable committee always exists. Via computer-assisted search, we have shown that this conjecture holds for small numbers of voters and candidates ($m + n \le 14$). However, as we have seen, existing voting rules seem incapable of finding such a committee, which makes this question very tantalizing.

The next open question is whether a stable lottery exists for ordinal preferences over committees that result from more general cardinal preferences. Our ``holy grail'' is to have {\em one algorithm} that finds a stable lottery using ordinal preferences over committees, that works for a wide range of underlying cardinal utilities.  One immediate extension of our work will be to the case where the utility of a voter is {\em additive} in the set of candidates in the committee. An obvious approach for attacking this problem is to extend Lemma~\ref{lem:pm} to the case where $Y_i \sim $ Bernoulli$(s_i,p_i)$, but this will require new ideas. 

Finally, our algorithm for $L$-stable lotteries has running time that depends on $m^L$. Though this is not an issue for {\sc Ranking}, it would be good to remove this dependence for {\sc approval set} and develop a poly-time algorithm for any $L$,  or show that this is not possible.

%% file: ms.bbl
\begin{thebibliography}{10}

\bibitem{AroraHK}
{\sc Arora, S., Hazan, E., and Kale, S.}
\newblock The multiplicative weights update method: a meta-algorithm and
  applications.
\newblock {\em Theory of Computing 8}, 6 (2012), 121--164.

\bibitem{Brill}
{\sc Aziz, H., Brill, M., Conitzer, V., Elkind, E., Freeman, R., and Walsh, T.}
\newblock Justified representation in approval-based committee voting.
\newblock {\em Social Choice and Welfare 48}, 2 (2017), 461--485.

\bibitem{PJR2018}
{\sc Aziz, H., Elkind, E., Huang, S., Lackner, M., Fernandez, L.~S., and
  Skowron, P.}
\newblock On the complexity of extended and proportional justified
  representation.
\newblock In {\em AAAI Conference on Artificial Intelligence (AAAI)\/} (2018).

\bibitem{Moulin}
{\sc Bogomolnaia, A., and Moulin, H.}
\newblock A new solution to the random assignment problem.
\newblock {\em Journal of Economic Theory 100}, 2 (2001), 295 -- 328.

\bibitem{Brams2007}
{\sc Brams, S.~J., Kilgour, D.~M., and Sanver, M.~R.}
\newblock A minimax procedure for electing committees.
\newblock {\em Public Choice 132}, 3 (2007), 401--420.

\bibitem{Chamberlain}
{\sc Chamberlin, J.~R., and Courant, P.~N.}
\newblock Representative deliberations and representative decisions:
  Proportional representation and the borda rule.
\newblock {\em The American Political Science Review 77}, 3 (1983), 718--733.

\bibitem{Droop}
{\sc Droop, H.~R.}
\newblock On methods of electing representatives.
\newblock {\em Journal of the Statistical Society of London 44}, 2 (1881),
  141--202.

\bibitem{ElkindFSS17}
{\sc Elkind, E., Faliszewski, P., Skowron, P., and Slinko, A.}
\newblock Properties of multiwinner voting rules.
\newblock {\em Social Choice and Welfare 48}, 3 (2017), 599--632.

\bibitem{Fain2016}
{\sc Fain, B., Goel, A., and Munagala, K.}
\newblock The core of the participatory budgeting problem.
\newblock In {\em Proceedings of the 12th Conference on Web and Internet
  Economics (WINE)\/} (2016), pp.~384--399.

\bibitem{FainMS18}
{\sc Fain, B., Munagala, K., and Shah, N.}
\newblock Fair allocation of indivisible public goods.
\newblock In {\em Proceedings of the 2018 {ACM} Conference on Economics and
  Computation (EC)\/} (2018), pp.~575--592.

\bibitem{FeigeV}
{\sc Feige, U., and Vondr\'ak, J.}
\newblock The submodular welfare problem with demand queries.
\newblock {\em Theory of Computing 6}, 11 (2010), 247--290.

\bibitem{lindahlCore}
{\sc Foley, D.~K.}
\newblock {L}indahl's solution and the core of an economy with public goods.
\newblock {\em Econometrica 38}, 1 (1970), 66--72.

\bibitem{Psomas}
{\sc Friedman, E., Gkatzelis, V., Psomas, C.~A., and Shenker, S.}
\newblock Fair and efficient memory sharing: Confronting free riders.
\newblock In {\em $33^{rd}$ AAAI Conference on Artificial Intelligence
  (AAAI)\/} (2019).

\bibitem{GandhiKS01}
{\sc Gandhi, R., Khuller, S., Parthasarathy, S., and Srinivasan, A.}
\newblock Dependent rounding and its applications to approximation algorithms.
\newblock {\em Journal of the ACM 53}, 3 (2006), 324--360.

\bibitem{HyllandZ}
{\sc Hylland, A., and Zeckhauser, R.}
\newblock The efficient allocation of individuals to positions.
\newblock {\em Journal of Political Economy 87}, 2 (1979), 293--314.

\bibitem{ROBUS}
{\sc Kunjir, M., Fain, B., Munagala, K., and Babu, S.}
\newblock {ROBUS:} fair cache allocation for data-parallel workloads.
\newblock In {\em Proceedings of the 2017 {ACM} International Conference on
  Management of Data ({SIGMOD})\/} (2017), pp.~219--234.

\bibitem{Lu2011}
{\sc Lu, T., and Boutilier, C.}
\newblock Budgeted social choice: From consensus to personalized decision
  making.
\newblock In {\em Proceedings of the 22nd International Joint Conference on
  Artificial Intelligence (IJCAI)\/} (2011), IJCAI'11, pp.~280--286.

\bibitem{Meir2008}
{\sc Meir, R., Procaccia, A.~D., Rosenschein, J.~S., and Zohar, A.}
\newblock Complexity of strategic behavior in multi-winner elections.
\newblock {\em J. Artif. Int. Res. 33}, 1 (Sept. 2008), 149--178.

\bibitem{Monroe}
{\sc Monroe, B.~L.}
\newblock Fully proportional representation.
\newblock {\em The American Political Science Review 89}, 4 (1995), 925--940.

\bibitem{coreConjectureCounter}
{\sc Muench, T.~J.}
\newblock The core and the lindahl equilibrium of an economy with a public
  good: An example.
\newblock {\em Journal of Economic Theory 4}, 2 (1972), 241 -- 255.

\bibitem{PanconesiS97}
{\sc Panconesi, A., and Srinivasan, A.}
\newblock Randomized distributed edge coloring via an extension of the
  chernoff-hoeffding bounds.
\newblock {\em {SIAM} J. Comput. 26}, 2 (1997), 350--368.

\bibitem{Procaccia2008}
{\sc Procaccia, A.~D., Rosenschein, J.~S., and Zohar, A.}
\newblock On the complexity of achieving proportional representation.
\newblock {\em Social Choice and Welfare 30}, 3 (Apr 2008), 353--362.

\bibitem{Sanchez}
{\sc S{\'a}nchez-Fern{\'a}ndez, L., Elkind, E., Lackner, M., Fern{\'a}ndez, N.,
  Fisteus, J.~A., Basanta~Val, P., and Skowron, P.}
\newblock Proportional justified representation.
\newblock In {\em Proceedings of the 31st AAAI Conference on Artificial
  Intelligence (AAAI)\/} (2017), pp.~670--676.

\bibitem{scarfCore}
{\sc Scarf, H.~E.}
\newblock The core of an n person game.
\newblock {\em Econometrica 35}, 1 (1967), pp. 50--69.

\bibitem{Thiele}
{\sc Thiele, T.~N.}
\newblock Om flerfoldsvalg.
\newblock In {\em Oversigt over det Kongelige Danske Videnskabernes Selskabs
  Forhandlinger\/} (1895), pp.~415--441.

\end{thebibliography}
